\documentclass{ifacconf}

\usepackage{graphicx}      
\usepackage{natbib}        

\usepackage{siunitx}
\usepackage{subcaption}
\usepackage{duckuments}
\usepackage{multirow}

\usepackage{tikz}

\usepackage{tkz-euclide}
\usetikzlibrary{arrows}

\newtheorem{proposition}{Proposition}
\AfterEndEnvironment{Proposition}{\noindent\ignorespaces}
\newtheorem{assumption}{Assumption}
\AfterEndEnvironment{Assumption}{\noindent\ignorespaces}
\newtheorem{remark}{Remark}
\AfterEndEnvironment{Remark}{\noindent\ignorespaces}

\usepackage{amsmath,amssymb,amsfonts}
\usepackage{bbm}

\newcommand{\tx}{\mathrm}

\newcommand{\pp}{\mathbf{p}}

\newcommand{\xx}{\mathbf{x}}
\newcommand{\bxx}{\Bar{\mathbf{x}}}
\newcommand{\BXX}{\Bar{\mathbf{X}}}
\newcommand{\bXX}{\mathbf{X}}
\newcommand{\XX}{\mathbb{X}}

\newcommand{\xxe}{\mathbf{x}_\tx{e}}
\newcommand{\xxh}{\mathbf{x}_\tx{h}}

\newcommand{\uu}{\mathbf{u}}
\newcommand{\buu}{\Bar{\mathbf{u}}}

\newcommand{\bUU}{\mathbf{U}}
\newcommand{\UU}{\mathbb{U}}

\newcommand{\uue}{\mathbf{u}_\tx{e}}
\newcommand{\uuh}{\mathbf{u}_\tx{h}}

\newcommand{\yy}{\mathbf{y}}
\newcommand{\YY}{\mathbb{Y}}
\newcommand{\bYY}{\mathbf{Y}}

\newcommand{\zz}{\mathbf{z}}

\newcommand{\bZZ}{\mathbf{Z}}

\newcommand{\RR}{\mathbb{R}}
\newcommand{\RRR}{\mathcal{R}}
\newcommand{\DD}{\mathbb{D}}
\newcommand{\DDD}{\mathcal{D}}
\newcommand{\EE}{\mathbb{E}}
\newcommand{\PP}{\mathbb{P}}
\newcommand{\bPP}{\mathbf{P}}
\newcommand{\PPP}{\mathcal{P}}
\newcommand{\FFF}{\mathcal{F}}

\newcommand{\AAA}{\mathcal{A}}

\newcommand{\KKK}{\mathcal{K}}

\newcommand{\NN}{\mathbb{N}}

\newcommand{\ocpNspace}{\!\!\!\!\!\!}

\DeclareMathOperator*{\argmin}{arg\,min}

\newcommand{\KL}{\tx{D}_\tx{KL}}

\newcommand{\bro}{\Bar{\rho}}

\newcommand{\quadsquared}{\quad \quad \quad \quad \quad \quad \quad \quad \quad \quad \quad \quad }

\begin{document}
\begin{frontmatter}

\title{Interactive Trajectory Planning with Learning-based Distributionally Robust Model Predictive Control and Markov Systems\thanksref{footnoteinfo}} 

\thanks[footnoteinfo]{This work was partially supported by the Wallenberg AI, Autonomous Systems and Software Program (WASP) funded by the Knut and Alice Wallenberg Foundation.}

\author[Chalmers,Volvo]{Erik Börve} 
\author[Chalmers]{Nikolce Murgovski} 
\author[Chalmers]{Morteza Haghir Chehreghani}
\author[Chalmers,Volvo]{Leo Laine}

\address[Chalmers]{Chalmers University of Technology, 
   Chalmersgatan 4, 412 96 Göteborg, Sweden (e-mail: \{borerik, first.last\}@chalmers.se).}
\address[Volvo]{Volvo Group Trucks Technology, 
   Gropegårdsgatan 2, 417 15 Göteborg, Sweden (e-mail: first.last@volvo.com)}

\begin{abstract}
We investigate interactive trajectory planning subject to uncertainty in the decisions of surrounding agents. To control the ego-agent, we aim to first learn the decision distribution and solve a \textit{Stochastic Model Predictive Control} (SMPC) problem. To account for errors in the learned distribution, we show that it is possible to utilize \textit{Probably Approximately Correct} (PAC) learning in combination with \textit{Distributionally Robust} (DR) optimization to obtain a solution which accounts for the errors induced by the learning model. The results indicate that our PAC learning-based DR-MPC framework provides a method to interpolate between a robust MPC and an omnipotent SMPC, based on the available number of samples.
\end{abstract}


\end{frontmatter}
\section{Introduction}
Trajectory planning is a crucial component of many robotic applications, such as manipulators, drones, and \sloppy Autonomous Vehicles (AVs) \citep{blackmore2011chance,borrelli2004collision,schouwenaars2001mixed}. Although much work has been dedicated towards this topic, the arguably most difficult challenge still remains: Safe and efficient planning in environments with exogenous agents, e.g humans. A major challenge to this end is the fact that the actions of the robot and humans can have an inherent influence on each other. Poor treatment of this interaction leads to an incorrect human motion model, which could significantly hinder performance and compromise safety by, e.g., causing dead-locks or collisions. Some examples of interactive trajectory planning are displayed in Fig. \ref{fig:scenario_examples}. Further, safety and efficiency can be conflicting objectives. A conservative planner may ensure high safety standards but could impair or even halt operations. Similarly, a more daring planner may attain higher performance, but could endanger nearby humans. Hence, an ideal robotic system should trade-off risk-minimization with performance-maximization while accurately assessing interactions between itself and humans. To this end, Model Predictive Control (MPC), combined with Machine Learning (ML), has gathered much attention in recent years. 
\begin{figure}[h!]
    \centering
    \begin{subfigure}{0.30\linewidth}
    \includegraphics[width=1\linewidth,trim={2.2cm 0.3cm 2.2cm 0.3cm},clip]{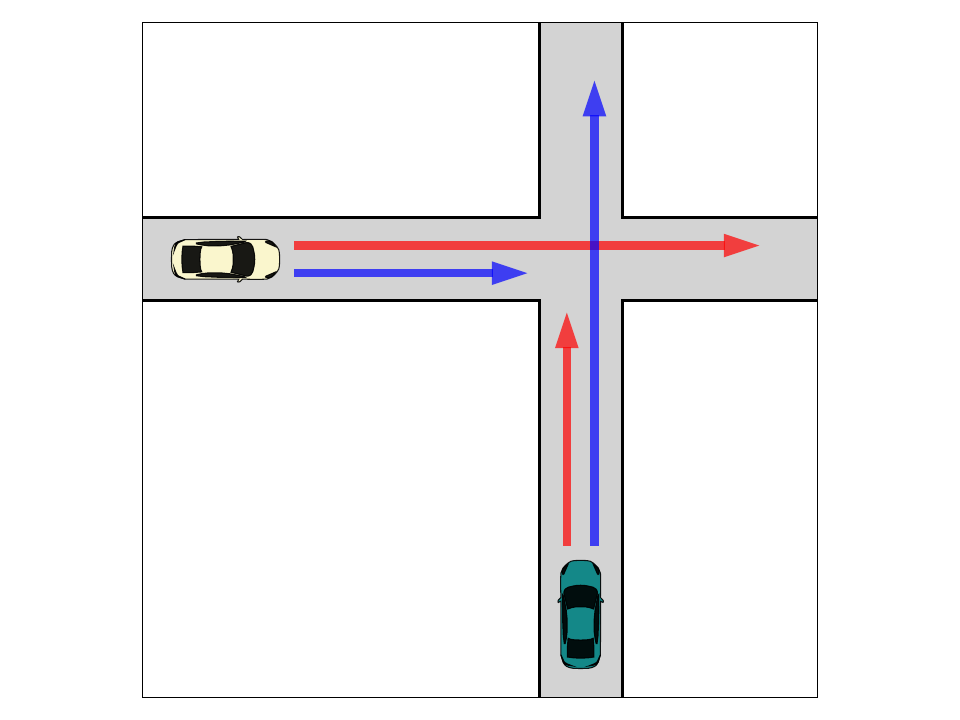}
    \end{subfigure}
    \begin{subfigure}{0.30\linewidth}
    \includegraphics[width=1\linewidth,trim={2.2cm 0.3cm 2.2cm 0.3cm},clip]{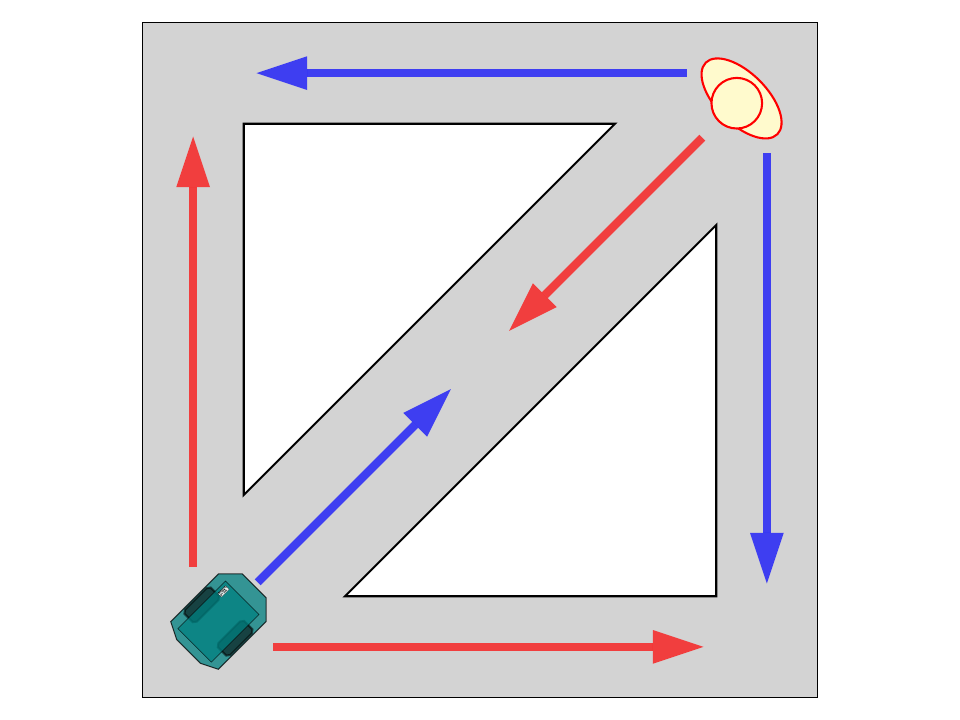}
    \end{subfigure}
    \begin{subfigure}{0.30\linewidth}
    \includegraphics[width=1\linewidth,trim={2.2cm 0.3cm 2.2cm 0.3cm},clip]{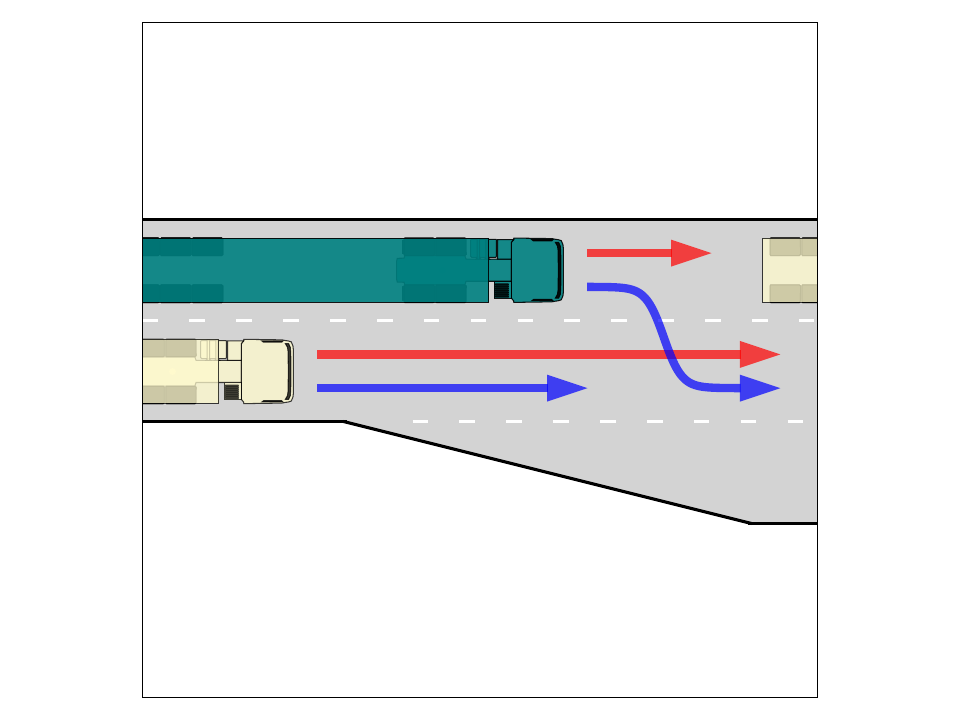}
    \end{subfigure}
    \caption{Interactive trajectory planning examples. The ego-agent (green) needs to negotiate with a human agent (yellow). Arrows indicate a discrete number of feasible alternatives in corresponding color.}
    \label{fig:scenario_examples}
\end{figure}
\subsection{Related Work}
A wide range of methods has been applied to this problem setting. In this paper, we will only discuss optimization-based methods and refer the interested reader to \cite{lavalle2006planning} for others. MPC optimizes performance metrics over a future horizon, subject to different model-based constraints, e.g., dynamics, actuator limitations, and collision avoidance. A compelling property of MPC is that it can be possible to establish rigorous guarantees on, e.g., optimality and constraint satisfaction, under the model assumptions. In dynamic environments, the collision avoidance constraints are heavily dependent on an accurate model of the surroundings. However, human movement is notoriously challenging to model, can be inherently uncertain and, in the case of interactions, depend on the robot's own motion. Recently, large-scale ML models have displayed impressive  performance for human trajectory predictions with uncertainty estimates \citep{salzmann2020trajectron++}. Such ML predictors have further been incorporated with Stochastic MPC (SMPC) to obtain a trade-off between performance and probabilistic collision avoidance satisfaction \citep{nair2022stochastic}. However, interactive planning is often lost, as large-scale ML predictors are difficult to integrate with the gradient-based optimizers used for MPC. Further, as the uncertainty quantification is heuristic, rigorous probabilistic constraint satisfaction is often lost. Some recent works has successfully integrated large-scale ML models \citep{borve2023interaction}, but without rigorously accounting for errors in the ML predictor. Rigorous constraint satisfaction has also been obtained with large-scale ML \citep{lindemann2023safe}. However, the guarantees hold only under assumptions that disregard interactions. Of particular interest to this work is \cite{schuurmans2023general}, that utilize Distributionally Robust Optimization (DRO) and treat humans as Markov Systems with a discrete, finite set of available decisions. A learning-based Distributionally Robust MPC (DR-MPC) is then utilized to obtain probabilistic guarantees on collision avoidance with respect to human decision making. A crucial component of this approach is the construction of ambiguity sets that contain the true human model with high probability. However, the utilized theoretical framework does not allow for ML models that are complex enough to describe interactions. 

\subsection{Outline and Contribution}
In this work, we extend the learning-based DR-MPC framework to allow for more complex models. In particular, we show that this extension allows for trajectory planning applications that consider both interactive decisions and provide rigorous uncertainty quantification. To this end, we leverage excess risk bounds from Probably Approximately Correct (PAC) learning to construct valid ambiguity sets for interactive human decision models. Finally, we demonstrate our learning-based DR-MPC in an interactive trajectory planning scenario and compare performance and constraint satisfaction with its robust and stochastic counterparts. In a broader sense, this work presents a step towards learning-based and efficient interactive trajectory planning with rigorous safety guarantees.

\section{Trajectory Planning Problem} \label{sec:MP_problem}
In this section, we first present the general SMPC-based trajectory planning problem. A key step involves formulating the problem over a scenario tree based on the possible human decisions. In a later section, we reformulate this problem as a learning-based Distributionally Robust Optimal Control Problem. To refer to the ego-agent we utilize subscript $\tx{e}$. For clarity, we consider a single human agent, referred to with subscript $\tx{h}$. However, the presented framework is applicable to any number of human agents.

\subsection{Ego- and Human-agent Dynamics}
We consider the ego-agent state $\xx_\tx{e} \in \XX_\tx{e} \subseteq \RR^{N_\tx{e,x}}$ and control actions $\uu_\tx{e} \in \UU_\tx{e} \subseteq \RR^{N_\tx{e,u}}$ with continuously differentiable, potentially non-linear, discrete-time dynamics $f_\tx{e}: \RR^{N_\tx{e,x}} \times \RR^{N_\tx{e,u}}\mapsto \RR^{N_\tx{e,x}}$. The human-agent is considered in a similar fashion with $\xx_\tx{h} \in \RR^{N_\tx{h,x}}$, $\uu_\tx{h} \in \RR^{N_\tx{h,u}}$, and $f_\tx{h}: \RR^{N_\tx{h,x}} \times \RR^{N_\tx{h,u}}\mapsto \RR^{N_\tx{h,x}}$. For notational convenience, we concatenate these attributes as $\bxx = [\xxe^\top,\xxh^\top]^\top$, $\buu = [\uue^\top,\uuh^\top]^\top$, $f = [f_\tx{e}^\top,f_\tx{h}^\top]^\top$.

Indeed, the aim of the MPC problem is to obtain the ego-agent control actions $\uu_\tx{e}$. Following the approach in, e.g \cite{schuurmans2023general,chen2022interactive}, the human control actions are characterized by a predefined set of control laws, each associated with a specific interactive decision. The decision-making process is modeled as a discrete random variable $\yy \in \YY = \{y_1, \dots, y_d\}$, where $d =|\YY|$. In, e.g., the road crossing scenario in Fig. \ref{fig:scenario_examples}, a human driver may, e.g., consider: $y_1$ corresponding to breaking, and $y_2$ corresponding to driving through the crossing. To consider interactions, the distribution of the decisions is dependent on the states of both the ego- and human-agent as,
\begin{equation} \label{eq:Ydist_def_general}
    \yy \sim \pp_\theta(\yy|\bxx) = \big[\PP[\yy=y_j \, | \, \bxx]\big]_{j=1}^d \,\, ,
\end{equation}
with absolute constant parameters $\theta$. The human control actions are then obtained by sampling a stochastic control law $\kappa: \RR^{N_\tx{h,x}} \times \YY \mapsto \RR^{N_\tx{h,u}}$. As in prior work, $f_\tx{h}$ and $\kappa$ are assumed to be available apriori while $\pp_\theta(\bxx)$ is considered unknown. In, e.g., the road crossing scenario of Fig. \ref{fig:scenario_examples}, $f_\tx{h}$ can, e.g., be chosen as a kinematic model, while $\kappa$ may be a set of human driver models, e.g., \cite{IDM}. While one might wish to also learn these functions, this problem setting already has the following attractive properties: 1.) Through the decision distribution, we explicitly account for the multimodality of interactive trajectory planning, in turn often reflecting the largest source of uncertainty, see, e.g., Fig. \ref{fig:scenario_examples}. 2.) With a tractable $\pp(\bxx)$ this model is well suited for gradient-based optimizers, enabling interactive planning. 3.) As we will soon show, we may obtain rigorous uncertainty quantification with a learning-based $\pp_\theta(\yy|\bxx)$. Considering uncertainty in $f_\tx{h}$ and $\kappa$ is a topic we consider of particular interest for future work.

\subsection{Stochastic Model Predictive Control on Scenario Trees}
\begin{figure}[t!]
    \centering
    \begin{tikzpicture}[
    roundnode/.style={circle, draw=black,line width=0.75pt, fill=white!40, minimum size=1pt,minimum size=1pt,scale=1.25},
    roundedbox/.style={draw,rounded corners,fill=gray!10,thick,inner sep=3pt,minimum width=8em},
    coverbox/.style={fill=white,},yscale=0.4,xscale=3
    ]

\draw[->,thick] (-0.25,-5) -- (2.5,-5) node[right] {$k$};
\draw[-,dotted,gray!30,thick] (0,4.5) -- (0,-4);
\draw[-,dotted,gray!30,thick] (1,4.5) -- (1,-4);
\draw[-,dotted,gray!30,thick] (2,4.5) -- (2,-4);

\draw [-,thick] (0,-5cm-8pt) node[below] {\small $0$}-- (0,-5cm+8pt) ;
\foreach \x in {1,2,...,2} {%
    \draw [-] (\x,-5cm-4pt) -- (\x,-5cm+4pt) node[below,yshift = -4.5pt] {\small $\x$};
}

\node[roundnode,label = {[yshift=-0.45cm]\footnotesize 0},minimum size = 1mm] at (0,0) (0) {};
\node[roundnode,label = {[yshift=-0.45cm]\footnotesize 1}] at (1,2) (1) {};
\node[roundnode,label = {[yshift=-0.45cm]\footnotesize 2}] at (1,-2) (2) {};


\node[roundnode,label = {[yshift=-0.45cm]\footnotesize 3}] at (2,4) (3) {};
\node[roundnode,label = {[yshift=-0.45cm]\footnotesize 4}] at (2,1) (4) {};
\node[roundnode,label = {[yshift=-0.45cm]\footnotesize 5}] at (2,-1) (5) {};
\node[roundnode,label = {[yshift=-0.45cm]\footnotesize 6}] at (2,-4) (6) {};


\draw[->,gray] (0) -- (1) node[midway,above,xshift=-5pt,yshift = 4pt] { \color{black} \footnotesize $\mathbb{P}({\yy_0 \! = \! y_1}|\bxx_0)$} ;
\draw[->,gray] (0) -- (2) node[midway,below,xshift=-5pt,yshift = -4pt] {\color{black} \footnotesize $\mathbb{P}({\yy_0 \! = \! y_2}|\bxx_0)$};


\draw[->,gray] (1) -- (3) node[midway,above,xshift=0pt,yshift = 6pt] { \color{black} \footnotesize $\mathbb{P}({\yy_1 \! = \! y_1}|\bxx_1)$} ;
\draw[->,gray] (1) -- (4) node[midway,below,xshift=0pt,yshift = -2pt] {\color{black} \footnotesize $\mathbb{P}({\yy_1 \! = \! y_2}|\bxx_1)$};

\draw[->,gray] (2) -- (5) node[midway,above,xshift=0pt,yshift = 2pt] { \color{black} \footnotesize$\mathbb{P}({\yy_2 \! = \! y_1}|\bxx_2)$} ;
\draw[->,gray] (2) -- (6) node[midway,below,xshift=0pt,yshift = -6pt] {\color{black} \footnotesize $\mathbb{P}({\yy_2 \! = \! y_2}|\bxx_2)$};

\node[text=black,font = \footnotesize,yshift=15pt,xshift=-2pt] at (0) {($\bxx_{0},\!\uu_{\tx{e},0}$)};

\node[text=black,font = \footnotesize,yshift=15pt,xshift=-2pt] at (1) {($\bxx_{1},\!\uu_{\tx{e},1}$)};
\node[text=black,font = \footnotesize,yshift=15pt,xshift=-2pt] at (2) {($\bxx_{2},\!\uu_{\tx{e},2}$)};

\end{tikzpicture}
    \caption{Example of a scenario tree for $|\YY|=2$.}
    \label{fig:scenario_tree_example}
\end{figure}
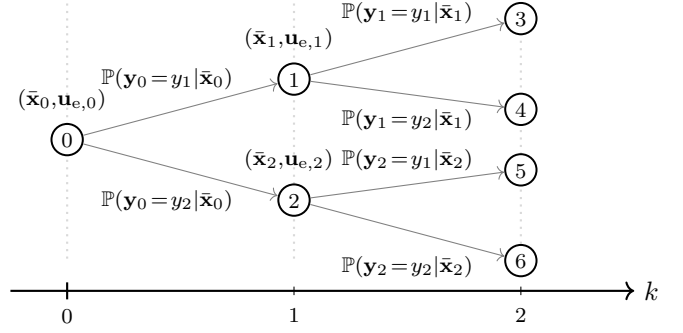
The SMPC problem is formulated over a finite, discrete-time horizon $k = 0, \dots, N$. As $|\YY|$ is finite, we can enumerate the human control actions and construct a directed scenario tree over the horizon. We will refer to nodes with the index $\iota \in \NN = \{0, \dots, N_\iota\}$. The tree starts from its root $\iota = 0$ and terminates at it's leaf nodes $\iota \in \NN_f$. To describe the parent of a node we utilize $\tx{Pr}(\iota)$, and similarly for the set of children $\tx{Ch}(\iota)$. Further, to ease notation, we introduce $\iota^+$ as a single node that is a direct descendant of $\iota$, i.e., $\iota = \tx{Pr}(\iota^+)$ and $\iota^+ \in \tx{Ch}(\iota)$. An example with $|\YY|=2$ is displayed in Fig. \ref{fig:scenario_tree_example}.

Each node contains the joint state $\bxx_\iota$ and each node, excluding leaf nodes, contains the ego-control actions $\uu_\tx{e,\iota}$ and a random variable $\yy_\iota$. Allowing a small notational abuse, we define $\tx{y}_{\iota^+}$ as the human decision at node $\iota$ responsible for the transition to node $\iota^+$. Hence, we may describe the joint state dynamics as follows,
\begin{equation}
    \bxx_{\iota^+} = f(\bxx_\iota, \buu_\iota | \yy_\iota = y_{\iota^+}), \, \forall \iota^+ \in \tx{Ch}(\iota).
\end{equation}
The dynamics are repeated for all nodes, excluding leaf nodes, producing all states in the scenario tree $\BXX = [\bxx_\iota]_{\forall \iota \in \NN}$. We may similarly define $\bXX_\tx{e} = [\xx_{\tx{e},\iota}]_{\forall \iota \in \NN}$, $\bUU_\tx{e} = [\uu_{\tx{e},\iota}]_{\forall \iota \in \NN \setminus \NN_f}$, and $\bYY = [\yy_\iota]_{\forall \iota \in \NN \setminus \NN_f}$. Ensuring collision avoidance amounts to ensuring that there exists a positive distance between the space occupied by each agent, for all nodes in the tree. We may describe this with the following scalar-valued constraints.
\begin{equation} \label{eq:dist_constraint}
    g(\bxx_\iota) = d_\tx{safe} -\tx{dist}(\xx_{\tx{e},\iota}, \xx_{\tx{h},\iota}) < 0, \, \forall \iota \in \NN.
\end{equation}
Here, the distance between the two agents is measured by $\tx{dist}: \RR^{N_\tx{e,x}} \times \RR^{N_\tx{h,x}} \mapsto \RR$ and $d_\tx{safe} \in \RR^+$ is an additional safety margin. Indeed, the states of the human agent are uncertain due to the random variables $\bYY$. To constrain the probability of encountering any collisions over the horizon, we may introduce a chance constraint as,
\begin{equation}
    \PP_{\bYY|\BXX} \left [ \cup_{\iota\in \NN\setminus\{0\}} \, g(\bxx_\iota) \geq 0 \, | \, \bxx_0 \right] \leq \varepsilon
\end{equation}
where $\varepsilon \in [0,1]$ is an accepted probability threshold.

Lastly, the SMPC problem considers a stage cost $\ell: \RR^{N_{\tx{e},x}} \times \RR^{N_{\tx{e},u}} \mapsto \RR$, terminal cost $\ell_f: \RR^{N_{\tx{e},x}} \mapsto \RR$ and a terminal set $\XX_{\tx{e},f}$. This finally yields the complete SMPC problem formulation over the scenario tree as,
\begin{subequations} \label{eq:SMPC}
    \begin{align} \label{eq:SMPC_objective}
    & \!\!\! \!\min_{\bUU_\tx{e}}&& \!\!\!\! \EE_{\bYY|\BXX} \left [\sum_{\iota \in \NN \setminus \NN_f} \ell(\xx_{\tx{e},\iota},\uu_{\tx{e},\iota}) + \sum_{\iota \in \NN_f} 
    \ell_f(\xx_{\tx{e},\iota}) \right] \\
    \label{eq:SMPC_dynamics}
    & \! \!\tx{s.t.}  && \ocpNspace \bxx_{\iota^+} \! = f(\bxx_\iota, \buu_\iota | \yy_\iota \!= y_{\iota^+}),\, \forall \iota^+  \! \! \in \!\tx{Ch}(\iota), \, \forall \iota \in \! \NN \setminus \NN_f \\ \label{eq:SMPC_dist}
    &&&\ocpNspace \yy_\iota \sim \pp_\theta(\yy_\iota|\xx_\iota), \; \forall \iota \in \NN \setminus \NN_f\\ \label{eq:SMPC_CC}
    &&&\ocpNspace \PP_{\bYY|\BXX} \left [\underset{\iota\in \NN\setminus\{0\}}{\bigcup} \, g(\bxx_\iota) \geq 0 \, \Big{|} \, \bxx_0 \right] \leq \varepsilon\\ \label{eq:SMPC_box_constraints}
    &&&\ocpNspace \xx_{\tx{e},\iota} \in \XX_\tx{e}, \; \forall \iota \in \NN \setminus \NN_f; \; \uu_{\tx{e},\iota} \in \UU_\tx{e},  \; \forall \iota \in \NN \setminus \NN_f\\ \label{eq:SMPC_boundary_constraints}
    &&&\ocpNspace \bxx_0 = \bxx(t); \; \xx_{\tx{e},\iota} \in \XX_{\tx{e},f}, \; \forall \iota \in \NN_f
    \end{align}
\end{subequations}
which includes: expected cost \eqref{eq:SMPC_objective}, dynamics \eqref{eq:SMPC_dynamics}, human decision distribution \eqref{eq:SMPC_dist}, chance constraint \eqref{eq:SMPC_CC}, ego-vehicle constraints \eqref{eq:SMPC_box_constraints} and boundary constraints \eqref{eq:SMPC_boundary_constraints}. Indeed, our aim is to solve the above problem for $\bUU_\tx{e}$ from the current traffic state $\bxx(t)$. One can obtain a tractable version of the above problem by reformulating \eqref{eq:SMPC_objective} and \eqref{eq:SMPC_CC}, as in, e.g., \cite{borve2025tight}, and learning the parameterized distribution $\pp_\theta$. However, to obtain rigorous guarantees we additionally need to quantify how well we are able to learn $\pp_\theta$. In the following sections, we propose a method to quantify this error utilizing ambiguity sets and further show that such ambiguity sets may be used to construct a tractable distributionally robust optimal control problem.

\section{Machine Learning Problem} \label{sec:ML_problem}
In the following section, we present the ML problem for learning the decision distribution. We further discuss generalization bounds on the learning scheme in the form of excess risk bounds. 

\subsection{Basic Definitions} \label{sec:basic_ML}
We consider a classification problem with random variables $\yy \in \YY = \{y_1,y_2,\dots,y_d\}$ and $\xx \in \XX \subseteq\mathbb{R}^{n_x}$ with a joint distribution $(\xx,\yy) \sim \DDD$. More precisely, we consider the joint distribution which is absolutely continuous with respect to $\xx$ as,
\begin{equation}
    \DDD = \bPP_\theta(\xx,\yy) = \pp(\xx)\pp_\theta(\yy|\xx)
\end{equation}
where $\pp_\theta(\yy|\xx) \in \PPP_\theta$ is the conditional label distribution and $\theta \in \RR^{n_x}$ are absolute constant parameters. In our MPC setting, $\yy$ are labels corresponding to a decision and $\xx$ are state-dependent features. Indeed, we aim to obtain an estimate of the conditional distribution by utilizing functions $\hat{\pp}_\theta(\yy|\xx) \in \FFF_\theta$. To this end, we define a risk measure utilizing the cross-entropy, i.e.,
\begin{equation} \label{eq:true_CE_risk}
    \RRR_\DDD(\hat{\pp}_\theta) = \EE_{\bPP_\theta} \big [-\log \hat{\bPP}_\theta \big] = H(\bPP_\theta,\hat{\bPP}_\theta)
\end{equation}
where $H(p,q)$ notes the cross-entropy between two distributions $(p, \,q)$ and $\hat{\bPP}_\theta = \pp(\xx)\hat{\pp}_\theta(\yy|\xx)$. With $n$ i.i.d samples from $\DDD$ we utilize empirical risk minimization (ERM) by constructing a training set $\DD = \{(x_i,y_i)\}_{i=1}^n$, and defining an empirical risk measure,
\begin{equation}
    \RRR_\DD(\hat{\pp}_\theta) = \frac{1}{n} \sum_{i=1}^n  - \log\big(\hat{\pp}_\theta(x_i,y_i) \big)
\end{equation}
where $\mathbf{1}_{\yy=y_i}$ is an indicator function for $\yy=y_i$. The ERM is then obtained by solving the following optimization problem,
\begin{equation} \label{eq:ERM_problem}
    \hat{\pp}_\theta = \argmin_{\pp} \{ \RRR_\DD(\pp) \; : \; \pp \in \FFF_\theta\}
\end{equation}
Naturally, since the estimate $\hat{\pp}_\theta$ is based on a finite sample of $\DDD$, the above problem will not return the true distribution. Indeed, we desire a result that is as close to the true distribution as possible. The following section provides a framework for measuring this discrepancy in terms of how well $\hat{\pp}_\theta$ minimizes the risk.
\subsection{Excess Risk Bounds} \label{sec:excess_risk_bounds}
Much work has been dedicated towards deriving bounds on the excess risk for many different machine learning algorithms. Bounds on the excess risk typically take the form of,
\begin{equation} \label{eq:excess_risk}
    \PP \left [\RRR_\DDD(\hat{\pp}_\theta)-\RRR_\DDD(\pp_\theta) \leq r(n,\alpha) \right] \geq 1- \alpha
\end{equation}
for some $r(n,\alpha) \in \RR^+$, depending on the number of samples $n$, and some probability threshold $\alpha \in [0,1]$. Similar to prior work in DR-MPC, we will assume uniform convergence, which we state more formally with the following assumption.

\begin{assumption}[Uniform Convergence] \label{as:uniform_convergence}
Consider the setting of Subsection \ref{sec:basic_ML} with risk measure $\RRR_\DDD$ of $\pp_\theta$ and empirical estimator $\RRR_\DD$. Uniform convergence implies,
\begin{equation}
    \lim_{n \to \infty } \PP \left [ \big | \min_{\pp \in \FFF_\theta}\RRR_\DD(\pp)-\RRR_\DDD(\pp_\theta) \big | > \epsilon \right] = 0.
\end{equation}
for any $\epsilon \geq 0$.
\end{assumption}
Bounds of the type \eqref{eq:excess_risk} are additionally challenging in the setting of Section \ref{sec:basic_ML} since the risk is unbounded, but may be obtained by imposing constraints on the random variable $\xx$ and $\theta$. Different approaches, e.g., VC-dimensions \citep{vapnik2015uniform}, Rademacher complexity \citep{bartlett2002rademacher} or PAC-Bayes bounds \citep{mcallester1998some}, yield different constraints with varying restrictiveness in different applications. Note that our approach is flexible with regard to the choice of bound and may be adapted based on $\PPP_\theta$ and $\FFF_\theta$. 
\begin{remark}
Assumption \ref{as:uniform_convergence} implies that the empirical risk of the ERM converges to the true risk of the true distribution. In the problem setting of Subsection \ref{sec:basic_ML}, this assumption is satisfied if the function class $\FFF_\theta$ is complex enough to include the true distribution $\pp_\theta$, i.e., $\PPP_\theta \subseteq \FFF_\theta$. This assumption is identical in principle to prior work, e.g., \cite{schuurmans2023general}, while allowing us to consider larger families for $\PPP_\theta$ and $\FFF_\theta$.
\end{remark}

\subsection{Example Problem Formulation} \label{sec:learning_problem_example}
To provide an early intuition for the learning setting, we describe a simple interactive distribution example that will be used later in the simulation study. We will consider a case where $\yy \in \{-1,1\}$ is conditionally Bernoulli distributed as,
\begin{equation} \label{eq:cond_dist_example}
    \pp_\theta(\yy=y|\xx=x) = \tx{Ber} \big(\sigma ( y\langle x,\theta \rangle)\big)
\end{equation}
where $\sigma(z) = \frac{1}{1+\exp(z)}$ and the parameters $\theta$ are unknown. Additionally, considering constraints on the random variable $\xx \in \mathbb{B}_\xx = \{\xx : ||\xx|| \leq B\}$ and the parameters $\theta \in \mathbb{B}_\theta = \{\theta  :  ||\theta|| \leq R \}$, we may utilize a Rademacher complexity bound on the excess risk. We summarize the result with the following proposition.
\begin{proposition}[Excess risk bound example] \label{prop:excess_risk}
Consider the above learning problem with definitions in Subsection \ref{sec:basic_ML}, $\xx \in \mathbb{B}_\XX$, ERM \eqref{eq:ERM_problem}, and $\FFF_\theta = \{\sigma(y \langle x, \theta \rangle) \! : \! ||\theta|| \leq R \}$. Under Assumption \ref{as:uniform_convergence}, we may bound the excess risk as,
\begin{equation*}
    \RRR_\DDD(\hat{\pp}_\theta)-\RRR_\DDD(\pp_\theta) \leq r(\alpha,n) = \frac{BR}{\sqrt{n}} \left ( 2 + \sqrt{2\log{\frac{2}{\alpha}}} \, \right )
\end{equation*}
which holds with probability $1-\alpha$.
\end{proposition}
\textit{Proof:} Application of \cite{bartlett2002rademacher}.

\section{Ambiguity sets from Excess Risk}
In this section, we will transform our empirical risk bound to a bound on the generalization error of the conditional distribution estimate. To this end, we will construct ambiguity sets that are suitable for obtaining a tractable DR-MPC problem. 

A natural approach to account for uncertainty in the distribution estimate is to consider a set of distributions in a certain neighborhood of the estimate. Such sets are often referred to as divergence-based ambiguity sets and typically take the form,
\begin{equation}
    \AAA_\epsilon = \left \{P : \mathfrak{d}(P,\hat{P}) \leq \epsilon \right \}
\end{equation}
where $\hat{P}$ is a probability estimate, $\mathfrak{d}(P,\hat{P}) \in \RR$ is a divergence measure, and $\epsilon \in \RR^+$ determines the size of the neighborhood \citep{rahimian2019distributionally}. Often, $\epsilon$ is selected to ensure probabilistic guarantees on the true distribution’s inclusion within the set. We will now show that cross-entropy-based risk measures are closely related to divergence-based ambiguity sets with the following proposition.
\begin{proposition}\label{prop:ambiguity_not_tractable} (Ambiguity Set from Excess Risk) \\
Consider the setting of Subsections \ref{sec:basic_ML}, \ref{sec:excess_risk_bounds} where the learning problem has a well defined bound on the excess risk $r(n,\alpha)$ as in \eqref{eq:excess_risk}. We may then consider the following ambiguity set,
    \begin{equation} \label{eq:ambiguity_not_tractable}
        \AAA_{\alpha,n} = \left \{\bPP : \KL \big( \bPP || \hat{\bPP}_\theta) \leq r(\alpha,n) \right \}
    \end{equation}
    where $\KL$ is the KL-divergence and $\bPP_\theta \in \AAA_{\alpha,n}$ with $1-\alpha$.
\end{proposition}
\textit{Proof:} From the definition of the cross entropy we have,
\begin{equation*}
    \RRR_\DDD(\hat{\pp}_\theta) = H(\bPP_\theta,\hat{\bPP}_\theta) = \KL\big (\bPP_\theta||\hat{\bPP}_\theta \big ) + H\left (\bPP_\theta \right)
\end{equation*}
where $H(p)= \EE_p[-\log p]$ is the entropy of a distribution $p$. Hence, 
\begin{align*}
    & \RRR_\DDD(\hat{\pp}_\theta) - \RRR_\DDD(\pp_\theta) = H(\bPP_\theta,\hat{\bPP}_\theta) - H(\bPP_\theta,\bPP_\theta)\\
    & =\KL\big (\bPP_\theta||\hat{\bPP}_\theta \big ) + H\left (\bPP_\theta\right) - H\left (\bPP_\theta \right) \\
    & =\KL\big (\bPP_\theta||\hat{\bPP}_\theta \big )
\end{align*}
Inserting in \eqref{eq:excess_risk} yields the desired result,
\begin{equation*}
    \PP \left [\KL\big (\bPP_\theta||\hat{\bPP}_\theta \big ) \leq r(n,\alpha) \right] \geq 1- \alpha
\end{equation*}
\hfill $\square$

Proposition \ref{prop:ambiguity_not_tractable} shows that the excess cross-entropy-based risk provides an ambiguity set for the joint distribution of $\xx$ and $\yy$. However, problem \eqref{eq:SMPC} treats the conditional distribution of $\yy$. Further, including \eqref{eq:ambiguity_not_tractable} does not produce a tractable SMPC problem as: (i) $\pp(\xx)$ may be unknown; (ii) The KL-divergence lacks tractable reformulations for many $\bPP_\theta(\xx,\yy)$ and $\hat{\bPP}_\theta(\xx,\yy)$. To address these issues, we construct an ambiguity set for the conditional distribution with the following proposition.

\begin{proposition}[Conditional Ambiguity Set]\label{prop:ambiguity_tractable}
    Consider the setting of Section \ref{sec:basic_ML} where the learning problem has a well defined bound on the excess risk $r(n,\alpha)$, as described in \eqref{eq:excess_risk}. For a given $x$ and an ERM of the conditional distribution $\hat{\pp}_\theta(\yy|x)$, we propose the following ambiguity set.
    \begin{equation} \label{eq:cond_ambiguity_set}
        \AAA_{\alpha,n}(x) = \left \{\pp \, : \, \KL \big( \pp || \hat{\pp}_\theta(\yy|x) \big) \leq \eta\big(r(\alpha,n)\big)  \right \}
    \end{equation}
    for which the following holds,
    \begin{equation*}
        \PP \left[\pp_\theta(\yy|x) \in \AAA_{\alpha,n}(x) \right] \geq \left(1-\frac{r(\alpha,n)}{\eta\big(r(\alpha,n)\big)}\right)(1-\alpha)
    \end{equation*} where $\eta \!: \![0,1] \!\mapsto\! [0,1]$ is increasing on its domain.
\end{proposition}
\textit{Proof:}
Expanding the KL-divergence between $\bPP_\theta$ and $\hat{\bPP}_\theta$ yields,
\begin{align*}
    \KL\big ( \bPP_\theta || \hat{\bPP}_\theta \big ) & = \int_\XX \pp(x) \sum_{i=1}^{d} \pp_\theta(y_i|x) \log \frac{\pp_\theta(y_i|x)}{\hat{\pp}_\theta(y_i|x)} \tx{d}x \\
    &= \EE_\xx [\KL \big( \pp_\theta(\yy|\xx) || \hat{\pp}_\theta(\yy|\xx) \big)]
\end{align*}
which integrates over continuous features $x$ and sums over the discrete labels $y_i$. For notational convenience, we introduce $\zz = \KL \big( \pp_\theta(\yy|\xx) || \hat{\pp}_\theta(\yy|\xx) \big)$ where $\EE_\xx[\zz] = \KL\big ( \bPP_\theta || \hat{\bPP}_\theta \big )$. From the definition of the KL divergence, we have $\zz \in [0,\infty)$ and may consider Markov's inequality as
\begin{equation} \label{eq:markovs_on_Z}
    \PP(\zz \geq t) \leq \frac{\EE_\xx[\zz]}{t}
\end{equation}
for some $t\geq 0$. Together with the result of Proposition \ref{prop:ambiguity_not_tractable}, we consider two dependent events,
\begin{equation*}
    A: \EE_\xx[\zz] \leq r(\alpha,n), \quad B: \zz < t.
\end{equation*}
By definition, we have $\PP(A) \in [0,1]$, and by Proposition~\ref{prop:ambiguity_not_tractable} we have $\PP(A) \geq (1-\alpha)$. Event B, where $\PP(B) \in [0,1]$, describes the complement of \eqref{eq:markovs_on_Z} with $\PP(B) > 1-\frac{\EE_\XX[\zz]}{t}$. Hence,
\begin{align*}
    &\PP(B|A) = \PP \big(\zz < t \,|\, \EE_\XX[\zz] \leq r(\alpha,n) \big)\\
    &> 1-\frac{\EE_\XX[\zz]}{t} \geq 1- \frac{r(\alpha,n)}{t}.
\end{align*}
We may pick any $t\geq 0$ and consider $t=\eta(r)$, where $\eta \!: \![0,1] \!\mapsto\! [0,1]$ is increasing on its domain. Bayes theorem then gives,
{\footnotesize
\begin{align*}
    \PP(B) = \frac{\PP(B|A)\PP(A)}{\PP(A|B)} \geq \PP(B|A)\PP(A) \geq \left(1-\frac{r(\alpha,n)}{\eta\big(r(\alpha,n)\big)}\right)(1-\alpha).
\end{align*}
}%
Inserting the definition of B and $\zz$, we obtain,
\begin{align*}
    &\PP \left [ \KL \big( \pp_\theta(\yy|x) || \hat{\pp}_\theta(\yy|x) \big) \leq \eta\big(r(\alpha,n) \big)\right ] \\
    &\geq \left(1-\frac{r(\alpha,n)}{\eta\big(r(\alpha,n)\big)}\right)(1-\alpha).
\end{align*}
\hfill $\square$





\section{Learning-based Distributionally Robust Model Predictive Control}
With an ambiguity set that quantifies the discrepancy between the true and estimated conditional distributions, our aim is now to solve \eqref{eq:SMPC}, accounting for the worst-case distribution in the ambiguity set. To this end, we will reformulate the objective \eqref{eq:SMPC_objective} and chance constraints \eqref{eq:SMPC_CC} using nested risk measures. Formally defining and proving all reformulations in this section requires an extensive review of the existing literature. Hence, we aim to discuss key properties of the reformulations and refer the interested reader to, e.g., \cite{schuurmans2023general,sopasakis2019risk,rahimian2019distributionally,shapiro2021lectures} for risk measures and DRO, and \cite{sopasakis2019risk,chen2022interactive} for nested risk. To further ease notation, we will use $(\cdot)_{\tx{Ch}(\iota)} = [(\cdot)_{\iota^+}]_{\forall \iota^+ \in \tx{Ch}(\iota)}$ to refer to a vector of all $(\cdot)_{\iota^+}$ children of a node $\iota$.

\subsection{Risk Measures in Distributionally Robust Optimization}
As is well explored in the literature, e.g., \cite{shapiro2021lectures}, the KL-divergence-based ambiguity sets can be expressed with an equivalent conic representation. In the case of \eqref{eq:cond_ambiguity_set} the conic representation becomes,
\begin{equation*}
        \AAA_{\alpha,n}(x) = \left \{ \pp \in \RR^d : E\pp + F \nu \preceq_{\KKK}b_{\hat{\theta}}(x) \right \}
\end{equation*}
where $\KKK$ is a closed, convex cone, $\nu$ are auxiliary variables, $E$, $F$ are matrices, and $b_{\hat{\theta}}(x)$ is a vector-valued function that contains $\hat{\pp}_\theta$. Importantly, as $\AAA_{n,\alpha}$ is convex and conic in $\pp$ and $\nu$, we can construct a coherent, conditional risk measure for a random quantity $\bZZ \in \RR^d$ as
\begin{equation} \label{eq:risk_measure}
    \rho_{|x}[\bZZ] = \max_{\pp \in \AAA_{n,\alpha}(x)} \EE_\pp [\bZZ] = \max_{\pp \in \AAA_{n,\alpha}(x)} \pp^\top \bZZ.
\end{equation}
Provided that strong duality holds, problem \eqref{eq:risk_measure} can equivalently be expressed as
\begin{equation} \label{eq:risk_measure_dual}
    \!\rho_{|x}[Z] = \min_{\lambda} \{\lambda^Tb_{\hat{\theta}}(x) : E^T \lambda = Z, F^T = 0, \lambda \succeq_{\KKK^*} 0 \}
\end{equation}
where $\KKK^*$ is the dual cone of $\KKK$.

\begin{remark}
    Observe that in the scenario tree setting of Section~\ref{sec:MP_problem}, \eqref{eq:risk_measure_dual} will be conditioned on the state $\bxx_\iota$ with $\bZZ_{\tx{Ch}(\iota)}$. Hence, the above risk measure describes the expected $\bZZ_{\tx{Ch}(\iota)}$, subject to the worst-case $\pp \in \AAA_{n,\alpha}(\bxx_\iota)$, which in turn describes the probability of transitions from $\iota$ to $\iota^+\in \tx{Ch}(\iota)$.
\end{remark}

\subsection{Nested Risk Measures over Scenario Trees}
The conditional risk measure \eqref{eq:risk_measure} can be utilized for a single $\iota$ with children $\tx{Ch}(\iota)$. To obtain a distributionally robust version of \eqref{eq:SMPC_objective} and \eqref{eq:SMPC_CC} we need a risk measure that treats all nodes in the scenario tree. To this end, we utilize the nested risk formulation of \cite{sopasakis2019risk}.

\subsubsection{Nested Risk Objective } Observe that expanding the expectation with respect to $\bYY$ in \eqref{eq:SMPC_objective} yields 
\begin{align} \label{eq:cost_expanded}
    &\ell(\xx_{\tx{e},0},\uu_{\tx{e},0}) + \EE_{\yy_0|\bxx_0}\Big[ \ell(\xx_{\tx{e},1},\uu_{\tx{e},1}) + \EE_{\yy_1|\bxx_1}\big[ \; \dots  \\ \nonumber
    &\ell(\xx_{\tx{e},\tx{Pr}(N_\iota)},\uu_{\tx{e},\tx{Pr}(N_\iota)}) + \EE_{\yy_{\tx{Pr}(N_\iota)}|\xx_{\tx{e},\tx{Pr}(N_\iota)}}[\ell_f(\xx_{\tx{e},N_\iota})]\dots\Big]
\end{align}
Given that $\ell$ and $\ell_f$ are lower semi-continuous and level-bounded over a closed set, we can formulate \eqref{eq:risk_measure} for each expectation to construct a version of $\eqref{eq:SMPC_objective}$ that treats the nested risk as,
\begin{equation*}
\begin{split}
    &\bro_{|\bxx_0} =  \rho_{|\bxx_0}\Big[\ell(\xx_{\tx{e},1},\uu_{\tx{e},1}) + \rho_{|\bxx_1} \big [\; \dots \\
    &\ell(\xx_{\tx{e},\tx{Pr}(N_\iota)},\uu_{\tx{e},\tx{Pr}(N_\iota)}) +  \rho_{|\bxx_{\tx{Pr}(N_\iota)}}[\ell_f(\xx_{\tx{e},N_\iota})] \big ] \dots  \Big]
\end{split}
\end{equation*}
where $\bro_{|(\cdot)}$ denotes the nested risk objective. As further detailed in \cite{sopasakis2019risk}, given that \eqref{eq:risk_measure} is coherent, the nested risk objective omits a tractable dual-reformulation. To this end, a crucial step is to observe that we may consider a recursive formulation of the conditional risk measure in the form
\begin{equation} \label{eq:risk_measure_recursive}
    \gamma_\iota = \rho_{|\bxx_\iota} \big [ \ell(\xx_{\tx{e},\iota^+},\uu_{\tx{e},\iota^+})_{\tx{Ch}(\iota)} + \gamma_{\tx{Ch}(\iota)} \big ].
\end{equation}
In the case of leaf nodes $\iota^+ \in \tx{Ch}({\iota}) \subseteq \NN_f$, we naturally have $\gamma_\iota = \rho_{|\bxx_{\iota}} \big[ \ell_f(x_{\iota^+})_{\tx{Ch}(\iota)} \big]$. Hence, each $\gamma_\iota$ measures the risk of all subsequent reachable nodes in the tree. Expressing \eqref{eq:risk_measure_recursive} for each node in the tree with \eqref{eq:risk_measure_dual}, we can obtain the following reformulation of the nested risk objective,

\begin{subequations}
    \begin{equation}
    \begin{aligned} \label{eq:nested_risk_objective}
    \min_{\bUU_\tx{e},\Gamma, \Lambda} \ell(\xx_{\tx{e},0},\uu_{\tx{e},0}) + \gamma_0 \quadsquared
    \end{aligned}
    \end{equation}
    \begin{equation} \label{eq:nested_risk_objective_constraints}
    \begin{split}
    &\ocpNspace  \tx{s.t.}  \; \lambda_\iota^\top b_{\hat{\theta}}(\bxx_\iota) \leq \gamma_\iota , \forall \iota \in \NN \setminus \NN_f \\
    & \ocpNspace E^\top\lambda_\iota \!= \!\gamma_{\tx{Ch}(\iota)}\! + \ell(\xx_{\tx{e},\iota^+},\uu_{\tx{e},\iota^+})_{\tx{Ch}(\iota)},\! \forall \iota \in \! \NN \setminus \! \tx{Pr}(\NN_f)\!\!\!\!\!\!\!\\
    &\ocpNspace  E^\top\lambda_\iota = \ell_f(\xx_{\tx{e}})_{\tx{Ch}(\iota)}, \forall \iota \in \tx{Pr}(\NN_f)\\
    &\ocpNspace F^\top \lambda_\iota = 0, \forall \iota \in \NN \setminus \NN_f;\; \lambda_\iota \succeq_{\KKK^*}0, \;\forall \iota \in \NN \setminus \NN_f
    \end{split}
    \end{equation}
\end{subequations}

where $\Gamma = [\gamma_\iota]_{\forall \iota \in \NN\setminus\NN_f}$ and $\Lambda = [\lambda_\iota]_{\forall \iota \in \NN\setminus\NN_f}$ gathers the additional optimization variables. The full proof is available in \cite{sopasakis2019risk} with an applied version in \cite{chen2022interactive}. To ease notation, we gather all constraints \eqref{eq:nested_risk_objective_constraints} with $\Omega_{\ell,\ell_f} = \{(\BXX,\bUU_\tx{e},\Gamma,\Lambda) : \eqref{eq:nested_risk_objective_constraints} \,| \, \bxx_0\}$.

\subsubsection{Nested Risk Constraints} We will now derive the corresponding distributionally robust chance constraints of \eqref{eq:SMPC_CC}, utilizing the same nested risk framework. Following \cite{borve2025tight}, we may express an outer approximation of the joint chance constraints over the scenario tree as,
\begin{align} \nonumber
    \!\!\! \PP_{\bYY|\BXX} \!\!\!\left [ \underset{\iota\in \NN \setminus \{0\}}{\bigcup} \!\!\!\! g(\bxx_\iota) \geq 0 \, \bigg| \, \bxx_0 \right] \!\!& \!\overset{(i)}{\leq} \! \EE_{\bYY|\BXX} \!\! \! \left [\sum_{\iota \in \NN \setminus \{0\}} \! \!\! \! \mathbbm{1}_{\RR^+} \big( g(\bxx_\iota) \big) \right] \\ \label{eq:CC_reform}
    & \!\overset{(ii)}{\leq} \! \! \EE_{\bYY|\BXX} \!\! \!\left [\sum_{\iota \in \NN\setminus \{0\}}  \! \!\! \! \sigma_{b,\beta}(g(\bxx_\iota)) \right]
\end{align}
where $\mathbbm{1}_{\RR^+}:\RR \mapsto \{0,1\}$ is an indicator function of the positive real line, and $\sigma_{b,\beta}(x) = b/(1+\exp(-\beta x))$ is a sigmoid function. Here, (i) follows from Boole's inequality and (ii) follows by picking $b$ and $\beta$ such that $\sigma_{b,\beta} \in \tx{Epi}\, \mathbbm{1}_{\RR^+}$, where $\tx{Epi}$ denotes the epigraph. One can directly observe that the result of \eqref{eq:CC_reform} has a similar structure to that of the objective. As $\sigma_{b,\beta} \circ g$ is continuous and level bounded over a closed set, the reformulation of the nested risk chance constraints follows in a similar manner to that of the nested risk objective. For brevity, we forego repeating the same arguments. Similar to the objective, we will introduce auxiliary variables, $\delta_\iota \in \RR$, $\Delta=[\delta_\iota]_{\forall\iota \in \NN \setminus \NN_f}$, to describe the corresponding version of  \eqref{eq:risk_measure_recursive} and $\mu_\iota$, $M = [\mu_\iota]_{\forall \iota \in \NN \setminus \NN_f}$ to describe the dual variables of the conditional risk measures \eqref{eq:risk_measure_dual}. To again ease notation, we introduce a set $\Omega_{\sigma_{b,\beta} \circ g}$ to gather the constraints of the corresponding nested risk chance constraints.

\subsubsection{DR-MPC Problem Formulation}
With a Distributionally robust formulation of \eqref{eq:SMPC_objective} and \eqref{eq:SMPC_CC} we can now construct the complete DR-MPC problem.

\begin{subequations}\label{eq:LB-DR-MPC}
    \begin{align} \label{eq:LB-DR-MPC_objective}
        &  \!\!\! \min_{\substack{\bUU_\tx{e},\Gamma, \Lambda\\ \Delta, M}}  && \!\!\!\! \! \ell(\xx_{\tx{e},0},\uu_{\tx{e},0}) + \gamma_0\\ \label{eq:LB-DR-MPC_dynamics}
        & \!\!\!\tx{s.t.}  && \ocpNspace \ocpNspace \bxx_{\iota^+} \! = f(\bxx_\iota, \buu_\iota | \yy_\iota \!= y_{\iota^+}), \forall \iota^+  \! \! \in \!\tx{Ch}(\iota), \forall \iota \in \! \NN \setminus \! \NN_f \!\!\!\!\\ \label{eq:LB-DR-MPC_CC}
        &&& \ocpNspace \ocpNspace \delta_0 \leq \varepsilon; \; (\BXX,\bUU_\tx{e},\Delta,M) \in \Omega_{\sigma_{b,\beta} \circ g} \\\label{eq:LB-DR-MPC_objective_duals}
        &&& \ocpNspace \ocpNspace (\BXX,\bUU_\tx{e},\Gamma,\Lambda) \in \Omega_{\ell, \ell_f} \\ \label{eq:LB-DR-MPC_ego_constraints}
        &&&\ocpNspace \ocpNspace\xx_{\tx{e},\iota} \in \XX_\tx{e}, \; \forall \iota \in \NN \setminus \NN_f; \; \uu_{\tx{e},\iota} \in \UU_\tx{e},  \; \forall \iota \in \NN \setminus \NN_f \\ \label{eq:LB-DR-MPC_ego_boundary}
        &&& \ocpNspace\ocpNspace\bxx_0 = \bxx(t); \; \xx_{\tx{e},\iota} \in \XX_{\tx{e},f}, \; \forall \iota \in \NN_f.
    \end{align}
\end{subequations}
Here, \eqref{eq:LB-DR-MPC_objective} describes the distributionally robust objective with corresponding constraints \eqref{eq:LB-DR-MPC_objective_duals}, \eqref{eq:LB-DR-MPC_dynamics} describes the dynamics, \eqref{eq:LB-DR-MPC_CC} describes the distributionally robust, joint chance constraints, \eqref{eq:LB-DR-MPC_ego_constraints} describes ego-agent limitations, and \eqref{eq:LB-DR-MPC_ego_boundary} describes the boundary conditions.  In a closed-loop setting, \eqref{eq:LB-DR-MPC} is solved at each discrete time, given the current traffic state $\bxx(t)$, with the ego-agent applying each $\uu_{\tx{e},0}$. In an open-loop setting, \eqref{eq:LB-DR-MPC} is solved once to obtain the control actions for the entire horizon $\bUU_\tx{e}$, applying the corresponding actions over the discrete horizon.

\section{Simulation Study}
\begin{figure}[t!]
    \centering
    \begin{subfigure}{0.49\linewidth}
    \includegraphics[width=1\linewidth,trim={2.2cm 0.4cm 2.2cm 0.4cm},clip]{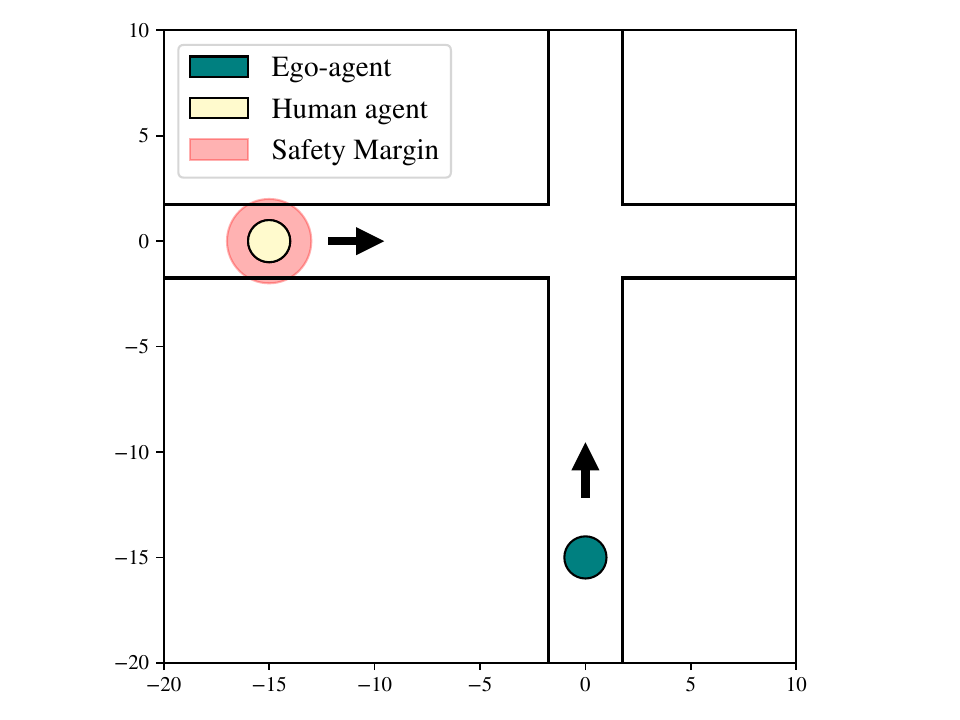}
    \caption{Initial simulation setup.}
    \label{fig:scene_two_way_intersection}
    \end{subfigure}
    \begin{subfigure}{0.49\linewidth}
    \includegraphics[width=1\linewidth,trim={0.35cm 0.8cm 0.2cm 0.35cm},clip]{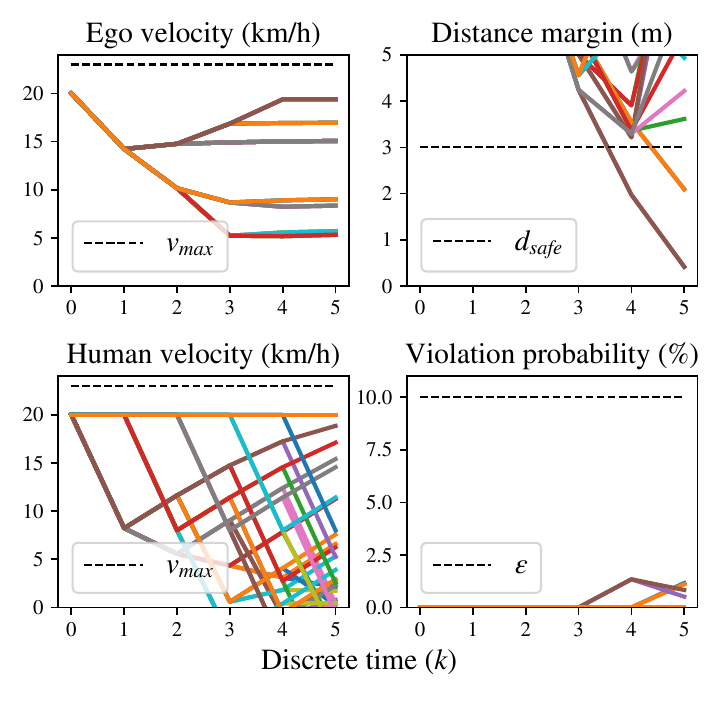}
    \caption{DR-MPC with $n\!=\!10^6$.}
    \label{fig:OCP_two_way_intersection}
    \end{subfigure}
    \caption{Road crossing case. Each color indicates a path through the scenario tree, from root- to leaf node.}
    \label{fig:two_way_intersection}
\end{figure}

We will now demonstrate our methods in an interactive trajectory planning example. To emphasize the properties of proposed method, we will consider a simple road crossing case between an Autonomous System (AS) and a human, see Fig. \ref{fig:scene_two_way_intersection}. To directly compare the result of the different optimization problems, we will investigate statistical results with the ego-agent adopting an open-loop control strategy. 

\subsubsection{Reference Controllers}
As a reference for the DR-MPC scheme, we will utilize two alternative MPC designs for solving problem \eqref{eq:SMPC}. First, a robust version (R-MPC), which does not utilize an estimate of $\pp_\theta$ and instead accounts for all possible $\xx_\tx{h}$ given $f_\tx{h}$ and $\kappa$, i.e., the collision avoidance constraints are not relaxed in any node, regardless of probability. Similarly, the objective is not weighted based on the probability of each respective node. Second, a stochastic version utilizing the ground truth distribution (GT-SMPC). This problem directly solves \eqref{eq:SMPC}, with the expected cost \eqref{eq:cost_expanded} and joint chance-constraint \eqref{eq:CC_reform}. More details of this reformulation is available in \cite{borve2025tight}. 

\subsubsection{Implementation} Both \eqref{eq:LB-DR-MPC} and \eqref{eq:ERM_problem} were formulated in $\tx{CasADi}$ \citep{casadi} and solved with $\tx{IPOPT}$ \citep{ipopt}. The simulations ran on a laptop equipped with a 12th Gen Intel(R) Core(TM) i7-12850HX CPU and 32 GB of RAM.

\subsection{Simulation Specifications}
Each agent $i \in \{\tx{e},\tx{h} \}$ is modeled as a one dimensional double integrator with states $\xx_i = [p_i,v_i]^\top$, control $\uu_i=[a_i]$, and discrete time dynamics $f_i(\xx_i,\uu_i) = [p_i+v_i \Delta t, v_i + a_i\Delta t]^\top$. In particular, $p_i$ reflects each agents position along their road, with the crossing placed in the origin. Both agents are initialized \SI{15}{\meter} away from the crossing at \SI{20}{\kilo \meter\per\hour}. The human agent may pick from two available decisions: $y_1$, corresponding to a breaking maneuver, and $y_2$, corresponding to tracking a reference $\xx_\tx{ref}$. The control law $\kappa$ is based on the Intelligent Driver Model \citep{IDM}, tuned for aggressive brake ($y_1$) and smooth reference tracking ($y_2$). For the decision distribution we consider a case with $\xx = -[p_\tx{e}/v_\tx{e},p_\tx{h}/v_\tx{h}]$, i.e., reflecting the signed time for the respective agent to reach the intersection, with $\pp_\theta(\yy|\xx)$ defined as in \eqref{eq:cond_dist_example} where $\theta = [3,3]$.

\subsubsection{Learner and Control Design}
\begin{table}[t!]
    \centering
    \caption{Shared Controller Parameters}
    \begin{tabular}{p{1.65cm}p{2.7cm}p{3cm}} \hline
         Parameter & Description & Values\\ \hline
         $\mathbf{Q}$ & State cost& $\tx{diag}(0,0.1)$\\
         $\mathbf{P}$& Terminal cost & $\tx{diag}(0,0.1)$\\
         $\mathbf{R}$& Control cost& $10$\\
         $\xx_{\tx{ref}}$& State reference & [0, \SI{20}{\kilo\meter\per\hour}] \\
         $\overline{\xx}$ &Upper state bound & [$\infty$, \SI{23}{\kilo\meter\per\hour}]\\
         $\underline{\xx}$ &Lower state bound & [$-\infty$, \SI{0}{\kilo\meter\per\hour}]\\
         $\overline{\uu}$ &Upper control bound & [\SI{5}{\meter\per\second\squared}]\\
         $\underline{\uu}$ &Lower control bound & [-\SI{5}{\meter\per\second\squared}]\\
         $r$ & Agent Radius & \SI{1}{\meter} \\
         $d_\tx{safe}$ & Safety margin & $2r+1$\SI{}{\meter}\\
         $\varepsilon$& Probability level& 0.1\\\hline
    \end{tabular}
    \label{tab:parameters_shared}
\end{table}
Indeed, the number of nodes, and further the number of variables scales exponentially with the length of the prediction horizon. In practice, scenario reduction strategies are often used to tackle this issue \citep{jacobsen2025combined}. In this study however, we aim to demonstrate the statistical properties, in reference to the ground-truth distribution. Hence, we opt for a relatively small $N=6$, with a relatively large $\Delta t= 1.0$ to balance computational complexity. All MPCs consider $\XX_\tx{e}= \XX_{\tx{e},f} = \{\xx : \underline{\xx} \leq \xx \leq \overline{\xx} \}$ and $\UU_{\tx{e}} = \{\uu : \underline{\uu} \leq \uu \leq \overline{\uu} \}$, with $\ell(\xxe,\uue) = \xxe^\top \mathbf{Q} \xxe + \uue \mathbf{R} \uue$ and $\ell_f(\xxe)=\xxe^T \mathbf{P} \xxe$. See Table \ref{tab:parameters_shared} for numerical values. As desired, this yields a challenging scenario, where the ego-vehicle is forced to take risk to improve performance. For the learner we will again consider the setting of Subsection \ref{sec:learning_problem_example}, with basic definitions in Subsection \ref{sec:basic_ML}. We will further consider Euclidean norms with $B = \sqrt{18}$, and $R = \sqrt{18}$, and $n$ i.i.d measurements of $x_i$ and $y_i$, assumed available offline. For the ambiguity set of Proposition \ref{prop:ambiguity_tractable}, we consider $\eta(r) = \sqrt{r}, r\in [0,1]$. As displayed in Table \ref{tab:results} we will investigate the solutions of the DR-MPC as the number of data points, and correspondingly as $\hat{\pp}_\theta$ and $r(\alpha,n)$, vary.


\subsection{Evaluation}
Fig. \ref{fig:OCP_two_way_intersection} provides a qualitative evaluation of the controller properties by displaying the predicted inter-vehicle distance (“Distance Margin”), the probability of violating the distance constraint (“Violation Probability”), and the velocity profiles of the respective vehicles. For a quantitative evaluation we may directly compute metrics by propagating the ground-truth distribution over the respective solutions $\bXX_\tx{e}$ and $\bUU_\tx{e}$. To this end, we consider: Expected Cost as the objective function subject to $\pp_\theta$; Crossing Rate as the rate of which the ego-vehicle crosses before the human; Stopping Rate as the rate of which the human crosses before the ego-vehicle, and Violation Rate as the rate of which the constraint $g$ is violated. The results are displayed in Table \ref{tab:results}. The R-MPC baseline attains violation free trajectory plans, but simultaneously obtains a remarkably high expected cost and low crossing rate. The omnipotent GT-SMPC manages to take the highest risk, and correspondingly achieve the highest performance. Recall that the GT-SMPC is not practically realizable as it is assumed to have access to the ground-truth distribution. The result for the DR-MPC can be observed to converge from the R-MPC, towards the GT-SMPC as the number of samples tends towards infinity, across all metrics. 
\begin{table}[t!]
    \centering
    \caption{Evaluation of Road Crossing Scenario.}
    \setlength{\tabcolsep}{3.6pt}
    \begin{tabular}{|c|c|c|c|c|c|}\hline
         & R-MPC& \multicolumn{3}{c|}{DR-MPC} &GT-SMPC\\\hline\hline
         $n$ & $-$ & $10^3$& $10^6$ & $10^9$ & $-$\\\hline 
         Expected Cost & 188.55 & 45.32 & 30.50 & 24.36 & 18.85\\ \hline
         Crossing Rate (\%) & 1.71& 50.00 & 50.31 & 55.90 & 58.76 \\ \hline
         Stopping Rate (\%) & 98.29 & $\sim$ 50.00 & 49.57 & 42.88 & 38.32\\ \hline
         Violation Rate (\%) & 0.00 & $\sim$0.00 & 0.12 & 1.22 & 2.92\\ \hline
    \end{tabular}
    \label{tab:results}
\end{table}

\subsection{Conclusions and Future Work}
The results indicate that the DR-MPC controllers provides a practically realizable method to interpolate between the R-MPC, and GT-SMPC methods based on the available number of samples $n$. The fact that the DR-MPC with $n=10^3$ provides a significant improvement over the robust controller, indicate that there exist some highly unlikely scenarios that significantly strain the performance of the R-MPC. However, we simultaneously observe that convergence towards the GT-SMPC with $n$ is slow. This could indicate that Propositions \ref{prop:excess_risk} and \ref{prop:ambiguity_tractable} are excessively conservative. Formulating tighter bounds for these propositions could be particularly interesting for practical applications. Further applications of this work could also consider different types of estimators, e.g., neural networks. We are additionally interested in rigorously investigating conditions for control guarantees, such as stability, recursive feasibility and probabilistic constraint satisfaction.




\begin{ack}
The authors thank Deepthi Pathare, Stefan Börjesson, Markus Gerdin, and Sten Elling Tingstad Jacobsen for insightful discussions concerning the research topic.
\end{ack}

\bibliography{ifacconf}

\appendix
\end{document}